\documentclass{llncs}

\usepackage[tight,footnotesize]{subfigure}
\usepackage{booktabs}
\usepackage{graphicx}
\usepackage{epstopdf}
\usepackage{amsmath}

\newtheorem{thrm}{Theorem}[section]

\begin{document}

\title{A Distributed Framework for Scalable Search over Encrypted Documents}

\maketitle

\begin{abstract}
Nowadays, huge amount of documents are increasingly transferred to the remote servers due to the appealing features of cloud computing. On the other hand, privacy and security of the sensitive information in untrusted cloud environment is a big concern. To alleviate such concerns, encryption of  sensitive data before its transfer to the cloud has become an important risk mitigation option. Encrypted storage provides protection at the expense of a significant increase in the data management complexity. For effective management, it is critical to provide efficient selective document retrieval capability on the encrypted collection. In fact, considerable amount of searchable symmetric encryption schemes have been designed in the literature to achieve this task. However, with the emergence of big data everywhere, available approaches are insufficient to address some crucial real-world problems such as scalability.

In this study, we focus on practical aspects of a secure keyword search mechanism over encrypted data on a real cloud infrastructure. 
First, we propose a provably secure distributed index along with a parallelizable retrieval technique that can easily scale to big data. 
Second, we integrate authorization into the search scheme to limit the information leakage in multi-user setting where users are allowed to access only particular documents. 
Third, we offer efficient updates on the distributed secure index. 
In addition, we conduct extensive empirical analysis on a real dataset to illustrate the efficiency of the proposed practical techniques.
\end{abstract}

\section{Introduction}

In recent years, advances in cloud computing has led to a rapid transition in information systems. Cloud services remove the burden of large scale data management in a cost-effective manner. Hence, it is quite common for individuals and organizations to outsource their documents. At the same time, storage of sensitive information in untrusted cloud environment raises serious privacy concerns. To resolve such concerns, one approach is to transfer documents in their encrypted form. Encrypted storage protects sensitive information at the expense of a significant reduction in the data computation capability. Among potential computations on the remote encrypted collection, selective retrieval is highly important for many mission critical tasks. In fact, a substantial amount of research effort has been invested to enable efficient execution of this operation.

Searchable symmetric encryption (SSE) schemes are the most common tools for searchable and secure cloud storage \cite{Chang,Curtmola,Goh,ccs,SDM}. 
Available SSE schemes enable selective document retrieval over encrypted collection. However, they do not address some practical problems of real systems such as scalability since they are mainly designed to run on a single server.  With the emergence of big data everywhere, scalability becomes a fundamental requirement for cloud systems. Fortunately, this challenge is resolved by an effective design principle that enforces distribution of both data and computation to multiple commodity hardwares. In fact, old storage systems are increasingly replaced with the ones that suits well to this distributed paradigm such as BigTable \cite{BigTable} and HBase \cite{Hbase}. Accordingly, it is critical to design a SSE scheme that can easily be distributed to many machines to handle very large amounts of documents. Another important practical aspect is to consider the access rights of distinct users during the retrieval. 
Data owners generally share limited amount of documents with other users.
To prevent excessive information leakage to the remote servers, integration of authorization into the search scheme is crucial. 

In this study, we propose a distributed secure index along with a parallelizable search mechanism. Specifically, we generate an inverted index on the document corpus which is subject to padding and secure encryption before its transfer to the cloud. This encrypted index is further vertically partitioned among multiple servers to enable simultaneous decryption of large index payloads during the search process. To our knowledge, this is the first effort in designing provably secure, distributed SSE scheme that can easily scale to big data. We also integrate authorization into the search to restrict the information leakage of the scheme according to user access rights. Finally, we propose a secure update mechanism on the distributed index which is a necessary functionality for practical cloud storage systems.
In summary, there are several notable contributions of this study:

{\bf Distributed SSE Scheme:} Current cloud infrastructures consist of many machines. To utilize their computational power to handle big data, we propose a secure index structure on top of a distributed key-value store known as HBase \cite{Hbase}. Proposed approach is highly scalable.  Our empirical evaluations indicate that search operation can be performed less than a second for approximately 1,200,000 documents on an HBase cluster of twelve machines.

{\bf Authorization-Aware Secure Keyword Search:} Almost all practical SSE schemes selectively leak information to the remote servers for efficiency. Although the leaked information varies among protocols, access pattern leakage is common. That is, untrusted server learns identifiers of the documents that are in the result set of an issued query without observing their contents. To restrict this leakage in multi-user setting, proposed scheme reveals only identifiers of the documents on which the user is authorized.

{\bf Secure Update on the Distributed Index:} Real document storage systems are highly dynamic. Some new documents are produced in time while some documents become outdated. To reflect these changes, we propose an effective update mechanism on the distributed index.   

The remainder of the paper is organized as follows. We review related work in Section \ref{Related}. We present distributed secure search scheme and its update mechanism in Section \ref{DistributedSearch} and Section \ref{Update} respectively. Then, in Section \ref{Security}, we analyze the security of the scheme. Finally, we report our experimental analysis in Section \ref{Experiments} and conclude in Section \ref{Conclusion}.

\section{Related Work}
\label{Related}
Over the years, various protocols and security definitions have been
proposed for searchable symmetric encryption (SSE). Optimal security is achieved by the ORAM model \cite{oblivous} of Goldreich et al. which does not leak any information. However, this model is impractical due to the excessive computational costs. Even the more efficient versions of ORAM  \cite{soda, revisited, pracORAM} are not practical enough to scale well for big data, especially in a multi-user setting.

As an alternative to ORAM, there are approaches (\cite{Chang,Curtmola,Goh,ccs,Practical,SDM}) which selectively leak information (e.g., access pattern) for more practical SSE schemes. The first of such approaches was provided in \cite{Practical}. Later on, Goh et al. proposed a security definition to formalize the security requirements of SSE in \cite{Goh}. Similarly, Chang et al. introduced a slightly stronger definition in \cite{Chang}. However, both definitions do not consider adaptive adversaries which could generate the queries according to the outcomes of previous queries. The shortcomings of \cite{Chang} and \cite{Goh} have been addressed in \cite{Curtmola}, where Curtmola et al. presented adaptive semantic security definition for SSE schemes.

The most computationally efficient SSE schemes that are compatible with the adaptive semantic security definition are presented in \cite{SDM}, \cite{ccs} and \cite{Kamara}. In \cite{SDM}, inverted bit vector indices are generated for each unique keyword. Then bit vectors are masked with a secure encryption scheme. These encrypted vectors can be stored and processed in our distributed framework. However, the construction of \cite{SDM} requires interaction with the user during the search process and its pure bit vector index structure is not efficient in terms of storage. To enable a single-round search through a more compact index along with an update mechanism, dynamic SSE scheme was proposed in \cite{ccs}. The approach of \cite{ccs} is based on a linked list structure. Specifically, for each keyword $w_i$,  identifiers of the documents that contain $w_i$ are randomly distributed to the cells of an array. Then, these cells are linked to each other with pointers. Finally, both pointers and cell contents are encrypted with a secure encryption scheme. Search mechanism of this construction necessitates sequential tracing on the encrypted list which hinders parallelizable search. To resolve the sequential tracing problem, a scheme that is based on red-black trees is proposed in \cite {Kamara}. In this construction, multiple-processors could apply decryption on distinct branches of the tree during the search but the whole tree resides in a single machine which is not compatible with the big data design principles. This tree keeps a leaf node for each document that consists of m-bit vector where m is the number of keywords in the corpus along with internal nodes that also consists of m-bit vectors. Hence, if millions of keywords exist in a very large document corpus, single machine, multi-processor setting will not be sufficient. In \cite{Cash}, a SSE scheme with support for conjunctive queries was proposed. Although this scheme could be parallelized in theory, it does not provide an update mechanism on the index which is fundamental for real world applications. Also, these SSE schemes leak access pattern without considering access-rights of distinct users.

\section{Secure Keyword Search}
\label{DistributedSearch}

In this part, we present our secure keyword search scheme that is based on a distributed index structure. In Section \ref{Index}, we describe the index construction mechanism. Then, in Section \ref{Search}, we focus on the search that is built on top of the index. Finally, in Section \ref{Authorization}, we provide an extension on the scheme for authorization integration.

\subsection{Secure Index Construction}
\label{Index}

The proposed secure index is formed in three main phases:

\vspace{1mm}
{\bf 1. Plain Index Generation:} In this phase, each keyword is associated with a set of documents. Specifically, suppose \{$D_1$, $...$, $D_n$\} is a set of documents with contents \{$W(D_1)$, $...$, $W(D_n)$\}, $id(D_j)$ is the identifier of $D_j$ and \{$w_1$, $...$, $w_z$\} represents the set of keywords. Then an inverted index \{$(w_1, L_{w_1})$, $...$, $(w_z, L_{w_z})$\} is formed such that 
$id(D_j) ~\in ~L_{w_i}$ if and only if $w_i ~\in ~W(D_j)$.

After inverted index construction, a plain payload is generated for each keyword in the form of a bit vector or a list. Payload type of a keyword is identified according to its frequency in the corpus with the goal of minimizing total storage cost. Here, both payload types should have a fixed size to hide keyword frequency. In fact, each bit vector consists of $n$ bits where $n$ represents the number of documents in the corpus. Similarly, each list consists of $\Upsilon$ document identifiers where $\Upsilon$ is a constant. During the construction, if the frequency of a keyword is more than $\Upsilon$, its payload is represented as a bit vector. Otherwise, it is represented as a list. In this setting, suppose $\Delta$ = \{$f_1$, $f_2$, ..., $f_z$\} represents the expected frequency distribution of the keywords, $|id|$ is the bit length of any identifier and $I(.)$ is an indicator function. Then expected storage cost denoted as $E(\Upsilon, \Delta, n)$ can be computed as follows: 

\vspace{-3mm}
{\footnotesize{
\begin{equation}
E(\Upsilon, \Delta, n) = \sum\limits_{i=1}^z I(f_i > \Upsilon) \cdot n + (1 - I(f_i > \Upsilon)) \cdot \Upsilon \cdot |id| 
\end{equation}}}
\vspace{-2mm}

To identify optimal $\Upsilon$ for minimal storage cost without leaking any information, we utilize Zipf's law \cite{zipf} which states that frequency of any keyword is inversely proportional to its frequency based rank in a natural language corpus. In our construction, payloads of top-t ranked keywords are put into bit vector form and the remaining ones are represented as a padded list as depicted in Figure \ref{Fig:Storage}.

\begin{figure}[!htb] 
\begin{center}
\includegraphics[width=55mm, height=25mm]{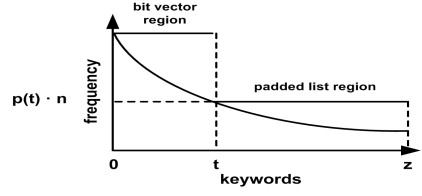}
\end{center}
\vspace{-7mm}
\caption{Payload Type Selection}
\label{Fig:Storage}
\end{figure}

More formally, suppose $H(N)$ represents $N^{th}$ harmonic number, $p(x; 1, N)$ represents the Zipfian distribution with exponent 1 such that $p(x) = x^{-1} \cdot H(N)^{-1}$ and $z$ is the number of keywords in the corpus. Then a padded list will have  $\Upsilon =  \lceil p(t) \cdot n \rceil$ members and the expected storage cost can be computed as follows:

{\footnotesize{
\begin{equation}
E(\Upsilon, p, n) = t \cdot n +  (z-t) \cdot p(t) \cdot  n  \cdot |id|
\end{equation}}}
\vspace{-4mm}

\begin{thrm}\label{thrm:storage}
Optimal member size for a padded list denoted as $\Upsilon$ with respect to the expected storage cost $E(\Upsilon, p, n)$ is $\lceil (|id| \cdot  H(z) \cdot z)^ {-0.5} \cdot n \rceil$.
\end{thrm}
\begin{proof}
{\footnotesize{
${arg\,min}_t~E(\Upsilon, p, n) \Rightarrow \\$

$\frac{dE(\Upsilon, p, n)}{dt} = \frac{d}{dt}[t \cdot n ~+ ~  (z-t) \cdot t^{-1} \cdot H(z)^{-1} \cdot n \cdot |id|] = 0\\$
$\\$
${arg\,min}_t~E(\Upsilon, p, n) = \sqrt{\frac{|id| ~\cdot ~z} {H(z)}}\\$
$\\$
$\Upsilon = \lceil p ({arg\,min}_t~E(\Upsilon, p, n)) \cdot n \rceil = \lceil (|id| \cdot  H(z) \cdot z)^ {-0.5} \cdot n \rceil$}}
\end{proof} 

Note that $\Upsilon$ is identified using public distribution. Hence it does not leak any information regarding the dataset. Once $\Upsilon$ is identified, a plain payload $P_{w_i}$ is generated for each keyword $w_i$. Specifically, suppose \{$(w_1, L_{w_1})$, $...$, $(w_z, L_{w_z})$\} is an inverted index where $L_{w_i}$ is a list of identifiers for the documents that contain $w_i$. Here, each identifier is an integer from $1$ to $n$ where n is the number of the documents in the corpus. Then $P_{w_i}$ is in the form of an n-bit vector if $|L_{w_i}| ~> ~\Upsilon$ such that $P_{w_i}[j] = 1$ for each $j \in L_{w_i}$ and  $P_{w_i}[j] = 0$ otherwise.  If $|L_{w_i}| ~\le ~\Upsilon$, $P_{w_i}$ is in the form a list such that $P_{w_i}$ is generated by inserting $\Upsilon - |L_{w_i}|$ fake identifiers (i.e., $id = 0$) to the original list $L_{w_i}$. 

\vspace{1mm}

{\bf 2. Index Encryption:} Plain index construction step results in an inverted index with fixed size bit vector or list payloads. In this step, this index is subject to encryption. Specifically, suppose \{$(w_1, P_{w_1})$, $...$, $(w_z, P_{w_z})$\} is a plain index, $\Phi_{K_t}$ and $\Psi_{K_p}$ are pseudo-random functions with secret keys $K_t$ and $K_p$, $O_S : \{0,1\}^\kappa ~\times~ \{0,1\}^{*} \mapsto \{0,1\}^{\kappa}$ is a random oracle\footnote{Cryptographic keyed hash functions such as HMAC-SHA256 could be utilized as a random oracle \cite{ccs}.}. Then encrypted index \{$(E_w(w_1), E_p(P_{w_1}))$, $...$, $(E_w(w_z), E_p(P_{w_z}))$\} is generated as follows:

\begin{list}{\labelitemi}{\leftmargin=0.2em}
\item Generate $E_w(w_i)$ such that $E_w(w_i)$ = $\Phi_{K_t}(w_i)$.
\item Generate a random oracle key $K(w_i)$ for the encryption of payload $P_{w_i}$ such that $K(w_i) = \Psi_{K_p}(w_i)$.
\item Construct payload encryption blocks.

1) Suppose $P_{w_i}$ is an n-bit vector and $\kappa$ is the output length of the random oracle. Then divide $P_{w_i}$ into $\kappa$ bit blocks $B_{w_i}^1$, ..., $B_{w_i}^\ell$ such that $\ell = \lceil n / \kappa \rceil$ and $B_{w_i}^j$ =  $P_{w_i}[(j - 1) \cdot \kappa + 1]$ ... $P_{w_i}[(j - 1) \cdot \kappa + \kappa]$. Here,  $P_{w_i}[k]$ is the $k^{th}$ bit of $P_{w_i}$ for $1 \le k \le n$ and zero otherwise.

2) Suppose $P_{w_i}$ is a list of identifiers with $\Upsilon$ elements such that $P_{w_i}$ = \{$id_{w_i}(1)$, ..., $id_{w_i}(\Upsilon)$\}, $|id|$ is the bit length of the identifiers and $\kappa$ is the output length of the random oracle. Then each encryption block can host $cp =  \lfloor \kappa ~/~ |id| \rfloor$ identifiers. Accordingly, divide $\Upsilon$ members into blocks $B_{w_i}^1$, ..., $B_{w_i}^\iota$ such that $\iota = \lceil \Upsilon / cp \rceil$ and $B_{w_i}^j$ =  $id_{w_i}((j - 1) \cdot cp + 1)$ ... $id_{w_i}((j - 1) \cdot cp + cp)$. Here, $id_{w_i}(k)$ is the $k^{th}$ identifier in $P_{w_i}$ for $1 \le k \le \Upsilon$ and zero otherwise.

\item Suppose $K(w_i)$ is the random oracle key, $B_{w_i}^1$, ..., $B_{w_i}^r$ are the encryption blocks for $P_{w_i}$, $O_S(K(w_i), j)$ denotes the output of random oracle $O_S$ with key $K(w_i)$ when it is applied on input $j$ and $\oplus$ represents the xor operator. Then generate encryption of $P_{w_i}$ denoted as $E_p(P_{w_i})$ as follows\footnote{This encryption construct is derived from the random oracle based
searchable encryption construction of \cite{ccs}.}:

\vspace{-3mm}
{\footnotesize{
\begin{align}
E_p(P_{w_i}) &= \{\pi_{w_i}[1], ... , \pi_{w_i}[r] \} \\
\pi_{w_i}[j] &= (B_{w_i}^j ~\oplus ~O_S(K(w_i), j),~ j)
\end{align}}}
\end{list}
\vspace{-1mm}

{\bf 3. Index Partitioning:} After the construction of encrypted index, it is transferred to the cloud. In cloud environment, encrypted index is split among multiple servers to parallelize the decryption of large index payloads during the keyword search phase. Specifically, encrypted payloads are divided into multiple regions which are further distributed to the region servers as depicted in Figure \ref{Fig:Partition}. 

\begin{figure}[!htb] 
\begin{center}
\includegraphics[width=80mm, height=35mm]{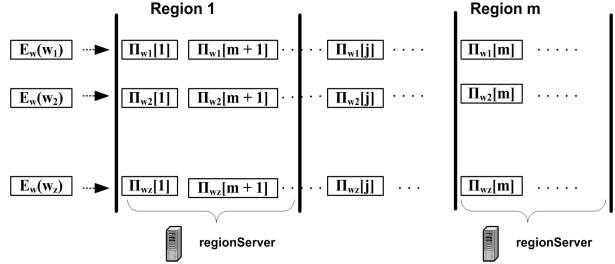}
\end{center}
\vspace{-6mm}
\caption{Index Partitioning}
\label{Fig:Partition}
\end{figure}

In this study, we utilize the features of a distributed key-value store known as HBase \cite{Hbase} for index partitioning. In HBase, key-value pairs can be split into distinct regions according to their keys. To construct such regions, region bound keys $R_1$, $R_2$, ..., $R_m$ are generated such that $R_1$ $<$ $R_2$ $<$ $...$ $<$ $R_m$. After bound initialization, key-value pairs are mapped to the regions according to relative position of their keys with respect to these bounds. Specifically, suppose $(k_i, v_i)$ is a pair that will be hosted in the store. Then it is placed into the $j^{th}$ region provided that $R_j ~\le~ k_i ~<~ R_{j+1}$. To form partitions on encrypted index \{$(E_w(w_1), E_p(P_{w_1}))$, $...$, $(E_w(w_z), E_p(P_{w_z}))$\}, we convert it into key-value pairs and map them into $m$ distinct regions as follows:

\begin{list}{\labelitemi}{\leftmargin=0.2em}
\item Suppose $\alpha$ is the bit length for the encrypted keywords such that $\alpha ~= ~|E_w(w_i)|$, $m$ is the number of regions, $0^{\alpha}$ represents a bit string of $\alpha$ zeros and $||$ is concatenation operator. Then generate region bound keys $R_1 = 1||0^{\alpha}$, $ R_2 = 2||0^{\alpha}$, $...$, $R_m = m||0^{\alpha}$.
\item Suppose $(E_w(w_i), E_p(P_{w_i}))$ is an encrypted keyword, payload pair where $E_p(P_{w_i})$ consists of $r$ encrypted blocks such that $E_p(P_{w_i})$ = \{$\pi_{w_i}[1]$, $\pi_{w_i}[2]$, $...$, $\pi_{w_i}[r]$\}. Then $m$ key-value pairs are constructed on  $(E_w(w_i), E_p(P_{w_i}))$. Specifically, $1||E_w(w_i)$, $2||E_w(w_i)$, $...$, $m||E_w(w_i)$ are the keys and value for key $j||E_w(w_i)$ is an encrypted block set \{$\pi_{w_i}[j]$, $\pi_{w_i}[j + m]$, $...$,  $\pi_{w_i}[j + m \cdot \lfloor r/m \rfloor]$\}. Here,  $\pi_{w_i}[k]$ is the $k^{th}$ block of $E_p(P_i)$ for $1 \le k \le r$ and empty otherwise.
\item Each key-value pair ($j||E_w(w_i)$, \{$\pi_{w_i}[j]$, $\pi_{w_i}[j + m]$, $...$,  $\pi_{w_i}[j + m \cdot \lfloor r/m \rfloor]$\}) is stored in  HBase. Value for row-key $j||E_w(w_i)$ is distributed to $\lceil r/m \rceil$ distinct columns. In this setting, pair with key $j||E_w(w_i)$ is stored in the $j^{th}$ region by the construction. This is because $j||E_w(w_i)$ lies between region bound keys $R_j = j||0^{\alpha}$ and $R_{j+1} = (j + 1)||0^{\alpha}$.
\end{list}

\subsection{Secure Search Mechanism}
\label{Search}

Once secure index is constructed and distributed to the cloud region servers, data users can perform search on them with the help of the secrets they own. 
During this setup, data owner also transfers the encrypted document collection  into the cloud along with the index. In this setting, encrypted collection denoted as $\{C_1, ..., C_n\}$ is obtained by applying secure encryption (e.g., AES in CTR mode of operation \cite{Goldwasser}) on the document collection. Suppose $Enc_{K_{coll}}$ represents the secure encryption with key $K_{coll}$ and $D_i$ is the $i^{th}$ document in the collection. Then $C_i = Enc_{K_{coll}}(D_i)$.  After the transfer of encrypted collection, it is stored in the distributed file system of the cloud service provider.

High-level overview of the search process is depicted in Figure \ref{Fig:Search}. Initially user generates a trapdoor for the keyword that he/she is interested in. Then he/she sends it to the cloud master server which directs it to the region servers. Each region server extracts the document identifiers for the issued trapdoor from the partial indices they have and send them back to the master. Finally, master combines the partial results and transfer the encrypted documents that are in the final result set to the user.

\begin{figure}[!htb] 
\begin{center}
\includegraphics[width=65mm, height=25mm]{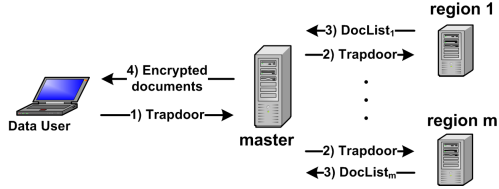}
\end{center}
\vspace{-5mm}
\caption{Search Mechanism}
\label{Fig:Search}
\end{figure}

Proposed search mechanism is in the form of well-known Map-Reduce programming model \cite{MapReduce} where region servers act as mappers and master server acts as a reducer. Here, map function extracts plain document identifiers from encrypted payload blocks with the help of the issued trapdoor. Then, reduce function combines the partial document identifier lists received from the mappers to finalize the search operation. Search process is performed in three main steps.

\vspace{1mm}
{\bf 1) Trapdoor Generation:} Trapdoor for retrieving documents that contain keyword $w_i$ denoted as $T_{w_i}$ is formed using secret keys $K_t$ and $K_p$. Specifically, let $\Phi_{K_t}$ and $\Psi_{K_p}$ be pseudo-random functions with secret keys $K_t$ and $K_p$. Then $T_{w_i}$ = ($E_w(w_i)$, $K(w_i)$) such that $E_w(w_i)$ = $\Phi_{K_t}(w_i)$ and $K(w_i)$ = $\Psi_{K_p}(w_i)$. Once $T_{w_i}$ is formed, it is sent to the cloud master server by the user.
\vspace{1mm}

{\bf 2) Region Search:} Upon reception of $T_{w_i}$ = ($E_w(w_i)$, $K(w_i)$), master directs it to the region servers to find the identifiers of the documents that contain $w_i$. Then the server that hosts $j^{th}$ region generates row-key $j||E_w(w_i)$. If ($j||E_w(w_i)$,\{$\pi_{w_i}[j]$, $\pi_{w_i}[j + m]$, $...$,  $\pi_{w_i}[j + m \cdot \lfloor r/m \rfloor]$\}) is a member of key-value store, server applies decryption on the encrypted payload blocks. In this setting, each block $\pi_{w_i}[k]$ is in the form of ($B_{w_i}^k ~\oplus ~O_S(K(w_i), k)$, $k$) by the construction where $O_S$ is a random oracle and $B_{w_i}^k$ represents the $k^{th}$ plain payload block for $w_i$. To obtain $B_{w_i}^k$, server performs decryption using $K(w_i)$ as follows:

\vspace{-2mm}
{\footnotesize{
\begin{equation*}
B_{w_i}^k = (B_{w_i}^k ~\oplus ~ O_S(K(w_i), k))  ~\oplus~ O_S(K(w_i), k)
\end{equation*}}}
\vspace{-4mm}

After decryption, server extracts document identifiers from each block $B_{w_i}^k$. Note that, each plain block is in the form of a bit vector or a list. If a block is in the bit vector form, identifiers are obtained from the bit locations with value one. If the block is in the list form, identifiers are obtained by decoding the encoded integers in it. Here, some decoded integers represent fake identifiers (i.e., $id = 0$) and they are eliminated by the region server at this phase. Once document identifiers are extracted for the issued trapdoor, region server forms identifier list for the $j^{th}$ region denoted as $doc_{w_i}(j)$ and transfers it to the master.

\vspace{1mm}

{\bf 3) Document Transfer:} In this step, master receives partial search results from each region server and merges them to finalize the search. Suppose, $doc_{w_i}$ represents the identifiers of the encrypted documents that contain $w_i$ and $doc_{w_i}(j)$ is the partial identifier list that is received from the $j^{th}$ region. Then $doc_{w_i}$ =  $doc_{w_i}(1)$ $~\cup~$ ... $~\cup$~ $doc_{w_i}(m)$. Finally encrypted documents, identifiers of which are included in $doc_{w_i}$, are sent back to the user.

\subsection{Authorization-Aware Keyword Search}
\label{Authorization}

Almost all efficient searchable symmetric encryption schemes leak some information for efficiency. Although this leakage varies among schemes, access pattern leakage is common. That is, untrusted
server learns which documents are included in the result set of
an issued trapdoor without observing their contents. More formally, suppose $D[q_i]$ represents the identifiers of the documents that are in the result set of query $q_i$. Then access pattern for $q_i$ denoted as $A_p(q_i)$ is a set such that  $A_p(q_i)$ =  $D[q_i]$. Unfortunately, access pattern leakage may subject to some adversarial analysis \cite{NDSS} and it is critical to restrict it in an efficient way \footnote{Although access pattern leakage can be eliminated completely with oblivious RAM, it is not practical enough to scale well for big data especially in multi-user setting.}.

In real-world data sharing, it is highly likely that data owner shares only subset of documents with other users through an authorization mechanism. 
Accordingly, if a data user issues a query $q_i$, he/she may have access to only some of the documents, identifiers of which are included in $A_p(q_i)$. However, available schemes do not consider access rights of the users and leak  $A_p(q_i)$ as it is during the search. To prevent excessive leakage in multi-user setting, we propose to integrate authorization into the process. Specifically, proposed scheme allows only the leakage of identifiers in $A_p(q_i)$ on which the user who issued the query is authorized.

In the context of this study, we form a secure search scheme with basic authorization as a first step toward search schemes with more sophisticated authorization techniques. Specifically, we utilize traditional file system access control lists (ACLs) \cite{Posix}. In ACLs, each file is included in a single access group and each user is assigned to multiple groups. Then, users are allowed to access only the files in their respective groups. More formally, suppose documents in the collection $D = \{D_1, ..., D_n\}$ are mapped to a group from the set $G = \{G_1, ...,G_p\}$ with function $g: D \mapsto G$. Similarly, user $U_i$ is assigned to a set of groups denoted as $G(U_i)$ such that $G(U_i) ~\subseteq ~G$. Then $U_i$ has access to document with identifier $id(D_j)$ if and only if $g(id(D_j)) ~\in ~G(U_i)$. Here, suppose $A_p(q_j)$ denotes the identifiers of the documents that are in the result set of query $q_j$ and $D(G_i)$ represents the identifiers of the documents that are in group $G_i$. Then restricted access pattern for $q_j$ and group $G_i$ denoted as $A_p(q_j, G_i)$ is a set such that $A_p(q_j, G_i)$ = $A_p(q_j) ~\cap ~D(G_i)$.

The main objective of authorization-aware keyword search is to restrict the access leakage according to the user access rights. Specifically, suppose $U_k$ issues a query $q_i$ and $G(U_k)$ denotes the groups that involves $U_k$. Then proposed scheme leaks only $A_p(q_i, G_\iota)$ for each $G_\iota \in G(U_k)$ instead of $A_p(q_i)$. To achieve this goal, we extend both secure index generation and search scheme of Sections \ref{Index} and \ref{Search}.

\vspace{1mm}
{\bf 1) Authorization-Aware Secure Index Construction:} In the basic index construction of Section \ref{Index}, each keyword $w_i$ is associated with encryption blocks $B_{w_i}^1$, ..., $B_{w_i}^r$, each of which contains the identifiers of the documents that include $w_i$. Then, these blocks are encrypted with the help of secret payload key $K_p$. In the extended version, we utilize multiple secret keys. Specifically, we generate group keys $K_{G_1}$, ..., $K_{G_p}$ and owner key $K_o$. The main goal of this design is to encrypt each document identifier with the key of its group. By this way, only the users that hold the corresponding group keys could generate a valid trapdoor for their decryption during the search phase. Here, owner key is formed for superusers such as data owner who access all the data to improve their search efficiency. After key generation, blocks are encrypted as follows:

\begin{list}{\labelitemi}{\leftmargin=0.2em}
\item If $B_{w_i}^j$ is a list block, it consists of concatenated document identifiers  $id_1||id_2 ... ||id_z$. Here, suppose $g(id_x)$ is the group of document with identifier $id_x$, $K_{G_i}$ is the secret key for group $G_i$, $\Psi_{K_{G_i}}$ is a pseudo-random function with key $K_{G_i}$,  $K(w_i, g(id_x))$ denotes random oracle key for pair $(w_i, g(id_x))$, $|id|$ is the bit length of an identifier and $e(V, \iota, \upsilon)$ is a function that extracts the block of bits between indices $\iota$ and $\upsilon$ from bit vector $V$. Then encryption of $B_{w_i}^j$ denoted as $\pi_{w_i}^j$ is formed as follows:

\vspace{-2mm}
{\footnotesize{
\begin{align*}
&K(w_i, g(id_x)) = \Psi_{K_{g(id_x)}}(w_i)\\
&\pi_{w_i}^j = (id_1 ~\oplus ~e(O_S(K(w_i, g(id_1)) , j), ~1, ~|id|) ~~~~~||~~~~~~ ...\\
&|| ~id_z \oplus~ e(O_S(K(w_i, g(id_z)) , j), (z-1) \cdot|id|+1, z \cdot |id|), ~j)
\end{align*}}}

\vspace{-2mm}
\item  If $B_{w_i}^j$ is a bit vector of length $k$, it consists of bits each of which represent a document identifier denoted as $id_{j,1}$, $id_{j,2}$, $...$, $id_{j,k}$ respectively. Here, suppose $G(B_{w_i}^j)$ is the set of groups for   the identifiers in block $B_{w_i}^j$ such that $G(B_{w_i}^j)$ =   $g(id_{j,1})$ $\cup~$ $...$  $~\cup$ $g(id_{j,k})$ where $g(id_{j,\rho})$ represents the group of document with identifier $id_{j,\rho}$. Then bit vector is initially encrypted with each group key $K_{G_{\iota}}$  where $G_{\iota} \in G(B_{w_i}^j)$ to form encrypted bit vector $\pi_{w_i}^j(G_{\iota})$:

\vspace{-4mm}
{\footnotesize{
\begin{align*}
&K(w_i, G_\iota) = \Psi_{K_{G_\iota}}(w_i)\\
&\pi_{w_i}^j(G_\iota) = B_{w_i}^j ~\oplus~ O(K(w_i, G_\iota), j)
\end{align*} }}
\vspace{-5mm}

Once encrypted bit vector is generated for each group in $G(B_{w_i}^j)$, final encrypted bit vector denoted as  $\pi_{w_i}^j$ is formed by group oriented bit selection. Specifically, suppose $id_{j,1}$, $...$, $id_{j,k}$ are the identifiers that is represented by bits 1, ..., k. Then $\rho^{th}$ bit of $\pi_{w_i}^{j}$ is set to $\rho^{th}$ bit of $\pi_{w_i}^j(g(id_{j, \rho}))$ where $g(id_{j, \rho})$ is the group of document with identifier $id_{j, \rho}$.
\end{list}

\vspace{-2mm}
After group encryptions, we also apply encryption with key $K_o$ on the blocks to speed-up the search for owner queries. Here, encryption of a block $B_{w_i}^j$ denoted as $\phi_{w_i}^j$ is performed in the same way as the block encryption of the basic index construction:

\vspace{-3mm}
{\footnotesize{
\begin{align*}
&K(w_i) = \Psi_{K_o}(w_i)\\
&\phi_{w_i}^j = B_{w_i}^j ~\oplus ~O(K(w_i), j)
\end{align*}}}
\vspace{-3mm}

In this setting, $(\pi_{w_i}^j, \phi_{w_i}^j, j)$ constitutes the final encryption for block  $B_{w_i}^j$. Apart from the distinction in payload block encryption, authorization-aware secure index construction is same as the basic secure index construction of Section \ref{Index}. For authorization-aware construction, documents in the collection are also encrypted according to their groups. Specifically, suppose $K_{G_\iota}^C$ is a secret collection key for group $G_\iota$, $g(id(D_i))$ denotes the group of $D_i$, $Enc_{K}$ is a secure encryption scheme with key $K$ and $C_i$ is the encrypted form of $D_i$. Then $C_i$ = $Enc_{K_{g(id(D_i))}^C} (D_i)$.

\vspace{1mm}
{\bf 2) Authorization-Aware Search Mechanism:}  

Authorization-aware search is an extension of the search mechanism that is presented in Section \ref{Search}. This extension integrates user-access rights into the search. For protocol execution, data owner shares secure index and group function $g : D \mapsto G$ which maps document identifiers to access groups with cloud service provider.  Service provider will further use this group information during the search to identify correct result-set against user trapdoors. Then owner shares keyword encryption key $K_t$, group keys $K_{G_j}$, $K_{G_j}^C$ for $G_j \in G(U_k)$ with user $U_k$. Here, $G(U_k)$ represents the groups that involves user $U_k$. Once necessary information is shared with the participants, trapdoor generation and region search on the cloud are performed as follows:

\vspace{1mm}
{\bf 1) Trapdoor Generation:} Trapdoor for retrieving documents that contain keyword $w_i$ denoted as $T_{w_i}$ is formed with secret keys $K_t$ and  $K_{G_j}$ for each $G_j \in G(U_k)$ = \{$G_{k,1}$, ..., $G_{k,f}$\}. Specifically, let $\Phi_{K_t}$ and $\Psi_{K}$ be pseudo-random functions with secret keys $K_t$ and $K$. Then $T_{w_i}$ = ($E_w(w_i)$, \{[$G_{k,1}$, $K(w_i, G_{k,1})$], $...$, [$G_{k, f}$, $K(w_i, G_{k, f})$]\}) where $E_w(w_i)$ = $\Phi_{K_t}(w_i)$ and $K(w_i, G_{k,j})$ = $\Psi_{K_{G_{k,j}}}(w_i)$. Once $T_{w_i}$ is formed, it is sent to the cloud master server by the user.

\vspace{1mm}
{\bf 2) Region Search:} Upon reception of $T_{w_i}$ = ($E_w(w_i)$, \{$[G_{k,1}, K(w_i, G_{k,1})]$, $...$, $[G_{k, f}, K(w_i, G_{k, f})]\})$, master directs it to the region servers.  Then each server locates the encrypted blocks using $E_w(w_i)$ as in the basic search scheme. Once they are located, decryption is applied. In this setting, each block consists of user and owner components $\pi_{w_i}^j$ and  $\phi_{w_i}^j$. Server decrypts $\pi_{w_i}^j$ for user trapdoors as follows:

\begin{list}{\labelitemi}{\leftmargin=0.2em}
\item If $\pi_{w_i}^j$ is a list block, it consists of concatenated encrypted document identifiers  $\omega_{id_{j,1}}||\omega_{id_{j,2}} ... ||\omega_{id_{j,z}}$. In this setting,  suppose $e(V, \iota, \upsilon)$ is a function that extracts the block of bits between indices $\iota$ and $\upsilon$ from bit vector $V$ and $|id|$ is the bit length of an identifier.  Then, for each oracle key $K(w_i, G_{k,\iota})$ in $T_{w_i}$, server performs the following operation:

\vspace{-3mm}
{\footnotesize{
\begin{align*}
id_{j, \rho}^{*} = \omega_{id_{j, \rho}} \oplus e(O_S(K(w_i, G_{k,\iota}), j), (\rho - 1) \cdot |id| + 1, \rho \cdot |id|)
\end{align*}}}

\vspace{-3mm}
Note that, legitimate document identifiers are integers between 1 and n and $g(id_{j, \rho})$ represents the group of the document with identifier $id_{j, \rho}$. In this setting, if $1 \le id_{j, \rho}^{*} \le n$ and $g(id_{j, \rho}^{*})$ = $G_{k,\iota}$, then $id_{j, \rho}^{*}$ is included in the search result set. Otherwise it is discarded by the server. It is clear that equality of $g(id_{j, \rho})$ and $G_{k,\iota}$ implies the equality of $id_{j, \rho}^{*}$ and $id_{j, \rho}$. On the other hand, decryption with a wrong group key will result in a random value.

\item  If $\pi_{w_i}^j$ is a bit vector block of length $k$, it consists of encrypted bits each of which represent document identifiers $id_{j,1}$, $...$, $id_{j,k}$. Here, suppose $G(B_{w_i}^j)$ is the set of groups for this block such that $G(B_{w_i}^j)$ = $g(id_{j, 1})$ $\cup$ $...$  $\cup$ $g(id_{j,k})$ and $G(U_k)$ is the set of groups that involves user $U_k$ who issued the trapdoor. Then bit vector is initially decrypted with each oracle key $K(w_i, G_{k, \iota})$ for $G_{k, \iota} ~\in ~(G(B_{w_i}^j)~ \cap ~G(U_k))$:

\vspace{-3mm}
{\footnotesize{
\begin{align*}
V_{w_i}^j(G_{k, \iota}) = \pi_{w_i}^j ~\oplus~ O(K(w_i, G_{k, \iota}), j)
\end{align*}}}
\vspace{-5mm}

Once bit vector is decrypted for each group in $G(B_{w_i}^j) \cap G(U_k)$, final  vector denoted as $V_{w_i}^j$ is formed by group oriented bit selection. Specifically, suppose $id_{j,1}$, ..., $id_{j,k}$ are the identifiers that are represented by bits 1, ..., k. Then $\rho^{th}$ bit of $V_{w_i}^j$ is set to $\rho^{th}$ bit of $V_{w_i}^j(g(id_{j,\rho}))$ if $g(id_{j,\rho}) \in ~G(U_k)$ and set to zero otherwise. Finally, if $\rho^{th}$ bit of $V_{w_i}^j$ is one,  $id_{j, \rho}$ is included in the search result. Note that, correct decryption of each bit can only be obtained with correct group key. Otherwise, decrypted bit will be random.
\end{list}

\vspace{-2mm}
After region servers identify document identifiers in their payloads, they transfer them to the master. Master merges the partial lists and sends the corresponding encrypted documents along with their groups to the user. In this setting, suppose $(C_i, G_\iota)$ is included in the result set. Then user decrypts $C_i$ with secret collection key $K_{G_\iota}^C$ to obtain its plain version. To speed-up the search for owner queries, encrypted blocks of authorization-aware secure index also contains owner blocks denoted as $\phi_{w_i}^j$. This enables the execution of basic scheme presented in Section \ref{Search} as it is.

\section{Secure Index Update}
\label{Update}
Real document storage systems are highly dynamic in their nature. Accordingly, data owner should be able to modify the encrypted collection that is hosted in the cloud. To achieve this goal, we extend our protocol and index structure to enable document deletion and addition.

\subsection{Document Deletion}
\label{Deletion}

Removal of a document from the cloud storage necessitates the elimination of the references against it that are included in the search index. To keep track of these references, we construct an update index. This new index
stores the cell addresses of the search index that are connected to a particular document and it is generated through three steps.

\vspace{1mm}
{\bf 1. Plain Index Construction:} During the search index setup, encryption blocks are constructed for each keyword in such a way that they consist of some slots, each of which contains or represents a particular document identifier. For any document $D_i$, `update index' stores these block and slot locations for the keywords that are included in $D_i$. More formally, suppose $D_i$ contains keyword set denoted as $W_i$ = \{$w_{i_1}$, ..., $w_{i_\ell}$\}, $E_w(w_{i_j})$ is the encrypted form of $w_{i_j}$, $name(D_i)$ is the unique filename of $D_i$, $block(w_{i_j}, D_i)$ and $slot(w_{i_j}, D_i)$ denote the block and slot order of $D_i$'s identifier in the encryption blocks of $w_{i_j}$. Then address list for $D_i$ denoted as $A(D_i)$ is a set such that $A(D_i)$ = \{$E_w(w_{i_1})$|| $block(w_{i_j}$,$D_i)$ || $slot$$(w_{i_j}$,$D_i)$, ..., $E_w(w_{i_\ell})$ || $block$$(w_{i_\ell}$,$D_i)$ || $slot$$(w_{i_\ell}$,$D_i)$\}. Finally, \{($name(D_1)$, $A(D_1)$), $...$, ($name(D_n)$, $A(D_n)$)\} constitutes plain `update index'.

\vspace{1mm}
{\bf Example \ref{Update}.1: } Consider the  plain search index blocks that are depicted in Figure \ref{Fig:PlainBlocks}. Encryption blocks for keyword $w_1$ are in in the form of bit vector while blocks for $w_2$ and $w_3$ are in the form of list. In this setting, address list for documents with identifiers $D_2$, $D_{56}$ and $D_{300}$ are as follows: $A(D_2)$ = \{$E_w(w_1)||1||2$, $E_w(w_2)||2||1$\}, $A(D_{56})$ = \{$E_w(w_2)||1||8$ \} and $A(D_{300})$ =  \{$E_w(w_1)||2||44$, $E_w(w_3)||2||5$\}.

\begin{figure}[!htb] 
\begin{center}
\includegraphics[width=75mm, height=25mm]{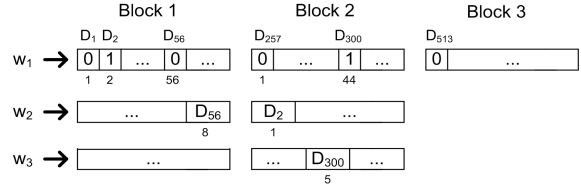}
\end{center}
\vspace{-6mm}
\caption{Sample Plain Search Index Blocks}
\label{Fig:PlainBlocks}
\end{figure}

{\bf 2. Document and Index Padding:} Note that, address list corresponding to document $D_i$  in plain `update index' contains $\ell$ cells provided that $D_i$ contains $\ell$ keywords. To hide the actual number of keywords included in a document, we apply padding on the document itself or its address list prior to their encryption. Suppose $avg_{|w|}$ is a parameter that indicates the unit bit length of a keyword, $|D_i|$ is the bit length of $D_i$, $A(D_i)$ is the address list for $D_i$ in the plain `update index'. Then $D_i$ is expected to contain $\xi(D_i)$ = $\lceil |D_i| / avg_{|w|} \rceil$ keywords.  In this setting, if $A(D_i)$ contains less than $\xi(D_i)$ keywords, we apply padding on the address list. Specifically, we insert  $|A(D_i)|$ - $\xi(D_i)$ fake cells (e.g., $E_w(fake)||-1||-1$) to $A(D_i)$. If $A(D_i)$ contains more than $\xi(D_i)$ keywords, we pad document itself with empty spaces until  $\xi(D_i)$ = $\lceil |D_i| / avg_{|w|} \rceil$.

\vspace{1mm}
{\bf 3. Index Encryption: } Suppose \{$(name(D_1), A(D_1))$, $...$, $(name(D_n), A(D_n))$\} is a padded plain `update index', $\Phi_{K_D}$ and $\Psi_{K_A}$ are pseudo-random functions with secret keys $K_D$ and $K_A$, $O_D$ is a random oracle. Then encrypted `update index' denoted as \{($E_D(D_1)$, $E_A(A_{D_1})$), $...$, ($E_D(D_n)$, $E_A(A_{D_n})$)\} is generated as follows:

\begin{list}{\labelitemi}{\leftmargin=0.2em}
\item  Generate $E_D(D_i)$ such that $E_D(D_i)$ = $\Phi_{K_D}(name(D_i))$.
\item Generate a random oracle key $K(D_i)$ for the address list encryption of $A(D_i)$ such that $K(D_i) = \Psi_{K_A}(name(D_i))$.

\item Suppose $K(D_i)$ is the random oracle key, $A_{D_i}^1$, ..., $A_{D_i}^r$ are members of $A(D_i)$, $O_D(K(D_i), j)$ denotes the output of random oracle $O_D$ with key $K(D_i)$ when it is applied on input $j$ and $\oplus$ represents  xor operator. Then generate encryption of $A(D_i)$ denoted as $E_A(A(D_i))$ as follows:

\vspace{-4mm}
{\footnotesize{
\begin{align*}
E_A(A(D_i)) &= \{\pi_{D_i}^1, ... , \pi_{D_i}^r \} \\
\pi_{D_i}^j &= (A_{D_i}^j ~\oplus ~O_D(K(D_i), j),~ j)
\end{align*}}}
\end{list}
\vspace{-2mm}

Once encrypted `update index' is generated, it is transferred to the cloud and stored in the key-value store. In this setting, for any document $D_i$, $(E_D(name(D_i))$ is a key and  $E_A(A(D_i))$ is the corresponding value. `Update index' enables data owner to delete documents from the remote servers. Deletion process is executed as follows:

\vspace{1mm}
{\bf 1. Deletion Token Generation: }  Token for deleting document with filename $name(D_i)$ denoted as $T_D(D_i)$ is formed using secret keys $K_D$ and $K_A$. Specifically, suppose $\Phi_{K_D}$ and $\Psi_{K_A}$ are pseudo-random functions with keys $K_D$ and $K_A$. Then $T_D(D_i)$ = ($E_D(D_i)$, $K(D_i)$) such that $E_D(D_i)$ = $\Phi_{K_D}$($name(D_i)$) and $K(D_i)$ = $\Psi_{K_A}$($name(D_i)$). Once $T_D(D_i)$ is formed, it is sent to the cloud master.

\vspace{1mm}

{\bf 2. Address Extraction: } Upon reception of $T_D(D_i)$ = ($E_D(D_i)$, $K(D_i)$), master fetches the value corresponding to row-key $E_D(D_i)$ from the key-value store. Suppose \{$\pi_{D_i}^1$, ... ,$\pi_{D_i}^r$\} is the value for 
$E_D(D_i)$. Then  master decrypts the encrypted addresses as follows:

\vspace{-1mm}
{\footnotesize{
\begin{equation*}
A_{D_i}^k =\pi_{D_i}^k ~\oplus~ O_D(K(D_i), ~k)
\end{equation*}}}
\vspace{-3mm}

After decryption, master obtains the address set for the search index cells which include reference to the document to be deleted. Note that, encrypted address list contains some fake entries by the construction. These fake entries are eliminated by the master at this phase. Final address list \{$A_{D_i}^1$, ..., $A_{D_i}^\rho$\} is utilized for adjusting the search index.

\vspace{1mm}

{\bf 3. Search Index Update: } Encrypted index blocks of the search index are distributed among $m$ distinct regions. To update search index,  master retrieves the necessary blocks from the region servers according to the address list \{$A_{D_i}^1$, ..., $A_{D_i}^\rho$\}. Note that, each address in the list is in the form of $E_w(w_j)$ || $block$$(w_j$,$D_i)$ || $slot$$(w_j$,$D_i)$. Here, $E_w(w_j)$ is the encrypted form of keyword $w_j$ which is included in $D_i$,  $block$$(w_j$, $D_i)$ and $slot$$(w_j$, $D_i)$ are the block and slot locations that holds the identifier for $D_i$.

By the search index construction, block with location $\upsilon_j^i$ = $block$$(w_j$,$D_i)$ is stored in the  $k^{th}$ region where $k =  \upsilon_j^i~ ~mod ~m$. Specifically, it is the value corresponding to  column $\upsilon_j^i$ of row-key $k||E_w(w_j)$. Master retrieves this value denoted as $\pi_{w_j}[\rho]$ from the corresponding region server and applies update on it according to its type as follows:

\begin{list}{\labelitemi}{\leftmargin=0.2em}
\item Suppose $\pi_{w_j}[\rho]$ is a bit-vector block. Then master updates its $\iota^{th}$ bit where $\iota = slot(w_j, id(D_i))$ such that  $\pi_{w_j}[\rho][\iota]$ = $\pi_{w_j}[\rho][\iota] \oplus 1$. Note that applied xor operation flips the bit to indicate the non-existence of keyword $w_j$ in $D_i$. Hence, further search results for $w_j$ will not include $id(D_i)$.
\item Suppose $\pi_{w_j}[\rho]$ is a list block. Then master extends the block with signal array denoted as $S_{w_j}[\rho]$ if there is no previous update on the block. Otherwise master updates the existing $S_{w_j}[\rho]$. In this setting, signal array is an $\eta$ bit vector where $\eta$ is the number of slots in a list block and each bit indicates the validity of the corresponding slot. If a deletion request is issued for the $\iota^{th}$ slot of the block where  $\iota = slot(w_j, D_i)$, $\iota^{th}$ bit is updated such that $S_{w_j}[\rho][\iota] = 1$ to invalidate this slot.

With signal array, query  evaluation on the search index is slightly different.  During the search, region servers check the signal array prior the decryption of the slots in a list block. If  $S_{w_j}[\rho][\iota] = 1$, then identifier that will be obtained from the decryption of $\iota^{th}$ slot is not included in the  result. 
\end{list}

\subsection{Document Addition}
\label{Addition}

To add a set of new documents to the encrypted collection, we propose a two-round protocol, overview of which is depicted in Figure \ref{Fig:Addition}. In the first round, data owner sends the number of slots that will be added to the search index corresponding to a set of keywords. According to these numbers, master transfers back the encrypted blocks with available slots. In the second round, data owner updates received blocks or generate some new blocks if available slots are not sufficient. Then, he forms the `update index' entries for the new documents. Finally, adjusted search index blocks and the new entries for the `update index' along with the encrypted documents are sent to the cloud.

\begin{figure}[!htb] 
\begin{center}
\includegraphics[height = 22mm, width = 80mm]{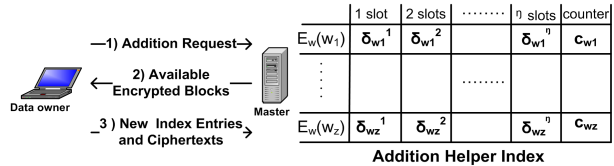}
\end{center}
\vspace{-6mm}
\caption{Document Addition Process}
\label{Fig:Addition}
\end{figure}

During the setup, search index consists of encrypted blocks which are in the bit-vector or list form. Since bit-vector form represents constant number of documents, we utilize only blocks of type list during the addition process. Although bit-vectors could be extended for addition in theory, it would be more costly in compare to lists. This is because, list insertion will modify only small subset of index while bit vector extension influences whole index. Hence, all blocks that will be inserted into the search index during the addition will be in the list form. To facilitate the addition, master server keeps an `addition helper index' as depicted in Figure \ref{Fig:Addition}. Note that, list blocks consist of slots, each of which contains a document identifier. Some of these slots are invalidated during the deletion process and they are available for further addition requests. In fact, helper index keeps track of these available slots along with the random oracle input counter for each encrypted keyword. 

\vspace{1mm}
{\bf Addition Helper Index: } Suppose $E_w(w_1)$, ..., $E_w(w_z)$ are the encrypted forms of keywords $w_1$, ..., $w_z$, $\eta$ is the number of slots in list blocks, $c_{w_i}$ denotes a counter for the blocks that are stored in the search index payload for $E_w(w_i)$, $\delta_{w_i}^j$ is a set of block locations with $j$ available slots. Then, value for helper index key  $E_w(w_i)$ denoted as $H_{w_i}$ is a pair such that $H_{w_i}$ = ([$\delta_{w_i}^1$, ..., $\delta_{w_i}^\eta$], $c_{w_i}$). During the search index setup, suppose $n_{w_i}$ encryption blocks are generated for keyword $w_i$. Then $c_{w_i}$ = $n_{w_i}$ and each $\delta_{w_i}^j$ for $1 \le j \le \eta$ is empty. Later, if $j$ slots of a list block for $E_w(w_i)$ are invalidated during the deletion process, location of this block is stored in  $\delta_{w_i}^j$. This indicates the availability of $j$ slots for the corresponding block.

\vspace{1mm}

Addition operation is performed in five main steps:

\vspace{1mm}
{\bf 1. Generation of addition request: } Suppose $D_U$ = \{$D_1^U$, ..., $D_\upsilon^U$\} is a collection of new documents with contents \{$W(D_1^U)$, ..., $W(D_\upsilon^U)\}$, \{$w_1^u$, ..., $w_\ell^u$\} is a set of all keywords that are included in the new documents, $id(D_i^U)$ is the identifier of $D_i^U$ and $\Phi_{K_t}$ is a pseudo-random function with secret key $K_t$. Then data owner generates an inverted index  \{($w_1^u$, $L(w_1^u)$), ..., ($w_\upsilon^u$, $L(w_\upsilon^u)$)\} such that $id(D_j^U) \in  L(w_i^u)$ if and only if $w_i^u \in W(D_j^U)$. Once inverted index is constructed, addition request  denoted as $T_A(D_U)$ is generated for the collection. Specifically, $T_A(D_U)$ = \{($E_w(w_1^u)$, $|L(w_1^u)|$), ..., ($E_w(w_\ell^u)$, $|L(w_\ell^u)|$)\} where $E_w(w_j^u) = \Phi_{K_t}(w_j^u)$ and $|L(w_j^u)|$ is the number of identifiers in $L(w_j^u)$ respectively.  Finally $T_A(D_U)$ is sent to the cloud master server.

\vspace{1mm}
{\bf 2. Block transfer:} Once master receives addition request, it utilizes `addition helper index' to locate blocks with available slots. If the number of slots is not sufficient, new blocks are generated by the user which we elaborate later. 

Suppose $(E_w(w_i^u)$, $|L(w_i^u)|)$ is included in $T_A(D_U)$ which implies that data owner needs  $|L(w_i^u)|$ slots for $E_w(w_i^u)$ to store new document identifiers. Then, master initially retrieves the value of $E_w(w_i^u)$ denoted as $H_{w_i^u}$ = ([$\delta_{w_i^u}^1$, ..., $\delta_{w_i^u}^\eta$], $c_{w_i^u}$) from the `addition helper index'. Then starting from  $\delta_{w_i^u}^\eta$ to  $\delta_{w_i^u}^1$, server extracts block locations from them until the total number of extracted slots reaches to $|L(w_i^u)|$. Note that, each block location that is extracted from $\delta_{w_i}^j$ contains $j$ slots. After this traversal, if the total number of slots does not reach to $|L(w_i^u)|$, data owner will generate new blocks. Specifically, if $r_i$ more slots are necessary, data owner will generate $\lceil r_i / \eta \rceil$ new blocks where $\eta$ denotes the number of slots in a single block.

Suppose \{$loc_{w_i^u}(1)$, ..., $loc_{w_i^u}(\rho_i)$\} are the block locations that are extracted by the helper index traversal and $m$ is the number of regions in the distributed search index.  Then these locations are removed from the helper index and corresponding blocks are retrieved from the search index. By the construction, block with location $loc_{w_i^u}(j)$ is stored in the $k^{th}$ region where $k = loc_{w_i^u}(j) ~mod~ m$. Master requests encrypted block with  location $loc_{w_i^u}(j)$ from the region server that hosts $k^{th}$ region. Finally retrieved blocks for $E_w(w_i^u)$ denoted as $\Delta(w_i^u)$ along with the counter in `addition helper index' denoted as $c_{w_i^u}$ are sent back to the data owner.

\vspace{1mm}
{\bf 3. Generation of new search index entries:} Suppose data owner receives block sets $\Delta(w_1^u)$, ..., $\Delta(w_\ell^u)$ against addition request. In this setting, $\Delta(w_i^u)$ consists of $\rho_i$ blocks, each of which contains some available slots to store new document identifiers. In this setting, each block in $\Delta(w_i^u)$ is in the form of $\pi_{w_i^u}[\iota]$ || $S_{w_i^u}[\iota]$ where $\pi_{w_i^u}[\iota]$ is the $\iota^{th}$ encrypted block in the search index for keyword $w_i^u$ and $S_{w_i^u}[\iota]$ is the signal array that keeps the invalid slot locations in it. Initially, data owner decrypts each block  $\pi_{w_i^u}[\iota]$ to obtain plain identifier list denoted as $id_{\iota,1}$|| ... || $id_{\iota, \eta}$. Then identifiers in invalid slots of the blocks are replaced with new identifiers according to $S_{w_i^u}[\iota]$. Specifically, suppose $L(w_i^u)$ is the identifier list for the new documents that contain $w_i^u$. Then, provided that $S_{w_i^u}[\iota][\upsilon] = 1$, a random member of  $L(w_i^u)$ is removed from it and  inserted into the $\upsilon^{th}$ slot of the plain block. If $L(w_i^u)$ is not empty after the identifier replacement on the retrieved blocks, data owner forms new blocks, slots of which are filled with the remaining member of  $L(w_i^u)$. Finally both old and new blocks which we call as update blocks are subject to encryption. 

Suppose  $c_{w_i^u}$  is the block counter for $w_i$, $id_{j,1}^u$||...|| $id_{j,\eta}^u$ is the content of $j^{th}$ plain block among the update blocks, $\Psi_{K_p}$ is a pseudo-random permutation with key $K_p$ and $O_S$($K(w_i^u)$, $j$) is the output of random oracle $O_S$ with key $K(w_i^u)$ when it is applied on input $j$. Then encrypted form of the block denoted as $\pi_{w_i^u}[j]$ is generated as follows\footnote{If authorization is enabled, each document identifier is encrypted through its group key as described in Section \ref{Authorization}.}:

\vspace{-2mm}
{\footnotesize{
\begin{align*}
&K(w_i^u) = \Psi_{K_p}(w_i^u) \\
&\pi_{w_i^u}[j] = ((id_{j,1}^u||...|| id_{j,\eta}^u) \oplus O_S(K(w_i^u), c_{w_i^u} + j), c_{w_i^u} + j)
\end{align*}}}
\vspace{-3mm}

After the encryption of the update blocks, their signal arrays are cleared to zero to indicate that all slots in them are valid and new block counter for $w_i$  is set to $c_{w_i^u} ~+ ~n_i$ where $n_i$ is the number of update blocks for $w_i$. 

\vspace{1mm}
{\bf 4. Generation of new update index entries:} Note that, we keep `update index' on the search index to enable further document deletion. This index stores the cell addresses of the search index that are connected to a particular document. Once identifiers corresponding to new documents are placed into the search blocks as explained in step 3, `update index' entries are generated according to these placements. More formally, suppose new document $D_i^u$ contains keyword set \{$w_{i_1}$, ..., $w_{i_z}$\}, $E_w(w_{i_j})$ is the encrypted form of $w_{i_j}$, $block(w_{i_j},D_i^u)$ and $slot(w_{i_j},D_i^u)$ denotes the block and slot locations that hosts the identifier of $D_i^u$ in the update blocks that are generated for $w_{i_j}$. Then address list for $D_i^u$ denoted as $A(D_i^u)$ is a set such that $A(D_i^u)$ = \{$E_w(w_{i_1})$|| $block(w_{i_j}$,$D_i^u)$ || $slot$$(w_{i_j}$,$D_i^u)$, ..., $E_w(w_{i_z})$ || $block$$(w_{i_z}$,$D_i^u)$ || $slot$$(w_{i_z}$,$D_i^u)$\}. In this setting, if update block is an old block that is received from the server, its location is its previous location. Otherwise, its location is equal to the block counter that is utilized during its encryption. Once address lists are formed for each new document, padding and encryption is applied on them as described  in Section \ref{Deletion}.

\vspace{1mm}
{\bf 5. Application of the updates:} At the final stage, new entries for `update index' and `search index' along with the encrypted documents are transferred to the cloud. Old search index blocks that are updated during the process are placed into their previous locations as depicted in Figure \ref{Fig:Distribution}. New search index blocks are uniformly distributed into the regions to assign similar load to each machine during the search process. In addition, block counter field of `addition helper index' is updated with the new block counters.

\begin{figure}[!htb] 
\begin{center}
\includegraphics[height = 19mm, width = 75mm]{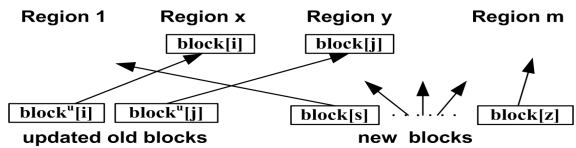}
\end{center}
\vspace{-6mm}
\caption{Distribution of the Search Blocks}
\label{Fig:Distribution}
\end{figure}

\section{Security Analysis}
\label{Security}
During the execution of the scheme, we assume that Bob who manages the cloud servers is semi-honest. As such, he follows the protocol as it is defined. However, he may try to infer private information about the documents he hosts. Over the years, many security definitions have been proposed for searchable symmetric encryption for semi-honest model. Among them, simulation based adaptive semantic security definition of \cite{Curtmola} is the widely accepted one in the literature. Later, it is
customized to work under random oracle model in \cite{ccs}  to be compatible with update operations. We adapt this definition for our construction in such a way that we consider user access-rights and keyword payload type while defining the legitimate information leakage. Adapted definition allows the leakage of  payload type, group-oriented access, addition, deletion and keyword patterns.

\vspace{1mm}
\textbf{Keyword Pattern ($\mathbf{KW_p}$):} Suppose $\{o_1, ..., o_\eta\}$ is a set of $\eta$ consecutive operations on the encrypted collection such that $o_i$ is a search, addition or deletion request. Note that each operation $o_i$ has a set of associated keywords denoted as $o_i^w$. Specifically, if $o_i$ is a search instance, it involves a single keyword such that $o_i^w = \{w_{i_1}\}$. If $o_i$ is a deletion, it involves a set of keywords that are included in the deleted document such that $o_i^w$ = \{$w_{i_1}$, ..., $w_{i_\upsilon}$\}. Finally, if $o_i$ is an addition, it involves a set of keywords that are included in the whole corpus of new documents such that $o_i^w$ = \{$w_{i_1}$, ... $w_{i_\varsigma}$\}. Then $KW_p$ is a function such that  $KW_p((i, \rho), (j,\ell)) = 1$ if $w_{i_\rho} = {w_{j_\ell}}$ and  $KW_p((i, \rho), (j,\ell)) = 0$ otherwise for $1 \le i, j \le \eta$, $1 \le \rho \le |o_i^w|$ and  $1 \le \ell \le |o_j^w|$.

To hide frequencies of the keywords that are included in the index in a storage optimal way, we utilize a fixed-size list or bit-vector payloads during the search index setup. If the number of documents that contain keyword $w_i$ is greater than a threshold $\Upsilon$, it is represented as a bit vector. Otherwise, it is represented as a list. To capture the leakage due to the payload type, we define a payload type pattern.

\vspace{1mm}
\textbf{Payload Type Pattern ($\mathbf{PT_p}$):} Suppose  \{$o_1$, ..., $o_\eta$\} is a set of $\eta$ consecutive operations, $o_i^w$ is a set of keywords that are included in operation $o_i$. Then $PT_p(i, j)$ = 1 if payload type for $w_{i_j}$ is bit-vector during the setup and  $PT_p(i, j)$= 0 otherwise where $1 \leq i \leq \eta$ and $1 \leq j \leq |o_i^w|$.

\vspace{1mm}
\textbf{{Group-Oriented Access Pattern ($\mathbf{A_p^g}$):}} Suppose $o_i$ is a search request for keyword $w_x$, $D(w_x)$ is the set of identifiers for the documents that contain keyword $w_x$, $g(id_j)$ denotes the group of document with identifier $id_j$, $block(id_j)$ and $slot(id_j)$ are the block and slot order in the search index payload for $w_x$ that hosts $id_j$. Then  $A_p(w_x, G_\iota)$ denotes a restricted access set such that [$id_j$, $block(id_j)$, $slot(id_j)$] $ ~\in~ A_p(w_x, G_\iota)$ if and only if $g(id_j) = G_\iota$ $\wedge$ $id_j \in  D(w_x)$. In this setting, suppose $G(U_i)$ = \{$G_{i,1}$, $...$ , $G_{i,f}$\} denotes the access groups of user $U_i$ who issued $i^{th}$ request. Then $A_p^g(o_i)$ = \{$(A_p(w_x, G_{i,1}),~G_{i,1})$, $...$, $(A_p(w_x, G_{i,f}), ~G_{i,f})$\}. 

\vspace{1mm}
\textbf{Deletion Pattern ($\mathbf{\beta_p}$):} Suppose $o_i$ is a deletion request for document $D_j$ which consists of keywords \{$w_{j_1}$, ...,$w_{j_\upsilon}$\}, $block(w_{j_\iota}, D_j)$ and $slot(w_{j_\iota}, D_j)$ denotes the block and slot order of $id(D_j)$ in the search index payload for $w_{j_\iota}$ and $|C_j|$ is the bit length of $D_j$'s encryption. Then  $L(o_i)$ = \{$block$($w_{j_1}$, $D_j$)||$slot$($w_{j_1}$, $D_j$), ..., $block$($w_{j_\upsilon}$)||$slot$($w_{j_\upsilon}$, $D_j$)\} is a location set and $\beta_p(o_i$) is a pair such that $\beta_p(o_i)$ = ($L(o_i)$, $|C_i|$).

\vspace{1mm}
\textbf{Addition Pattern ($\mathbf{\alpha_p}$):} Suppose $o_i$ is an addition request for a document collection \{$D_\iota$, ..., $D_\rho$\}, $|C_x|$ denotes the bit-length for the encrypted form of $D_x$, \{$w_{j_1}$, ...,$w_{j_\nu}$\} is a set of keywords that are included in the new corpus and $cnt(w_{j_\iota})$ denotes the number of documents in the new corpus that contain $w_{j_\iota}$. Then $\alpha_p(o_i)$ = (\{$|C_\iota|$, ..., $|C_\rho|$\},  \{$cnt(w_{j_1})$, ..., $cnt(w_{j_\nu})$\}).

Security definition based on these leakages which we call authorization-aware  adaptive semantic security for dynamic searchable symmetric encryption along with the security proof of the proposed scheme is provided in Appendix-A.

\section{Experimental Analysis}
\label{Experiments}

In this section, we provide an empirical analysis of the
proposed scheme. To perform our evaluation, we utilized a publicly available
dataset of real emails, namely Enron dataset \cite{Enron}. We selected all 30109 emails included in the sent-mail folder of all users as our experimental corpus.
Prior to index generation on the corpus, we applied Porter stemming algorithm \cite{Porter} on the e-mail contents to extract the roots of each word. After stemming and eliminating common words like `the', corpus consists of approximately 77000 keywords. We also generated a corpus of approximately 1,200,000 emails by data replication for scalability test. To simulate a real cloud environment, we formed a HBase cluster of twelve machines. Generated secure indices are distributed among these machines according to the proposed architecture.

\subsection{Search Evaluation}
\label{KeywordExperiment}
In this part, we evaluate the computationally efficiency of search scheme with distinct settings.  Storage type threshold ($\Upsilon$), dataset size ($n$), number of regions for index partitioning ($m$), and number of user groups ($|G(U)|$) are the parameters. Default values for these parameters are as follows:  $\Upsilon = 6$, $n = 30109$, $m = 1$, $|G(U)| = 1$. To investigate the influence of distinct parameters, we modified a single parameter at a time and used the default values for the others. 

In this study, we propose an inverted index where payloads are in the form of a encrypted bit vector or list during the setup. Payload type depends on a threshold $\Upsilon$ which is identified according to Theorem \ref{thrm:storage}.  This theorem assumes that frequencies of the keywords are distributed according to Zipfian distribution with exponent 1. It is clear that, it does not accurately capture the frequency distribution of the underlying dataset, but it is generally a close approximation. In fact, our empirical analysis indicate that default $\Upsilon$ (i.e., $\Upsilon$ = $6$) that is computed according to Theorem \ref{thrm:storage} provides considerable storage savings in compare to pure bit vector or list payloads as depicted in Figure \ref{Fig:Exp1}-a.

\begin{figure*}[!htb]
\centering
\subfigure[$\Upsilon$ vs. index storage]{
\includegraphics[height=32mm, width=39mm]{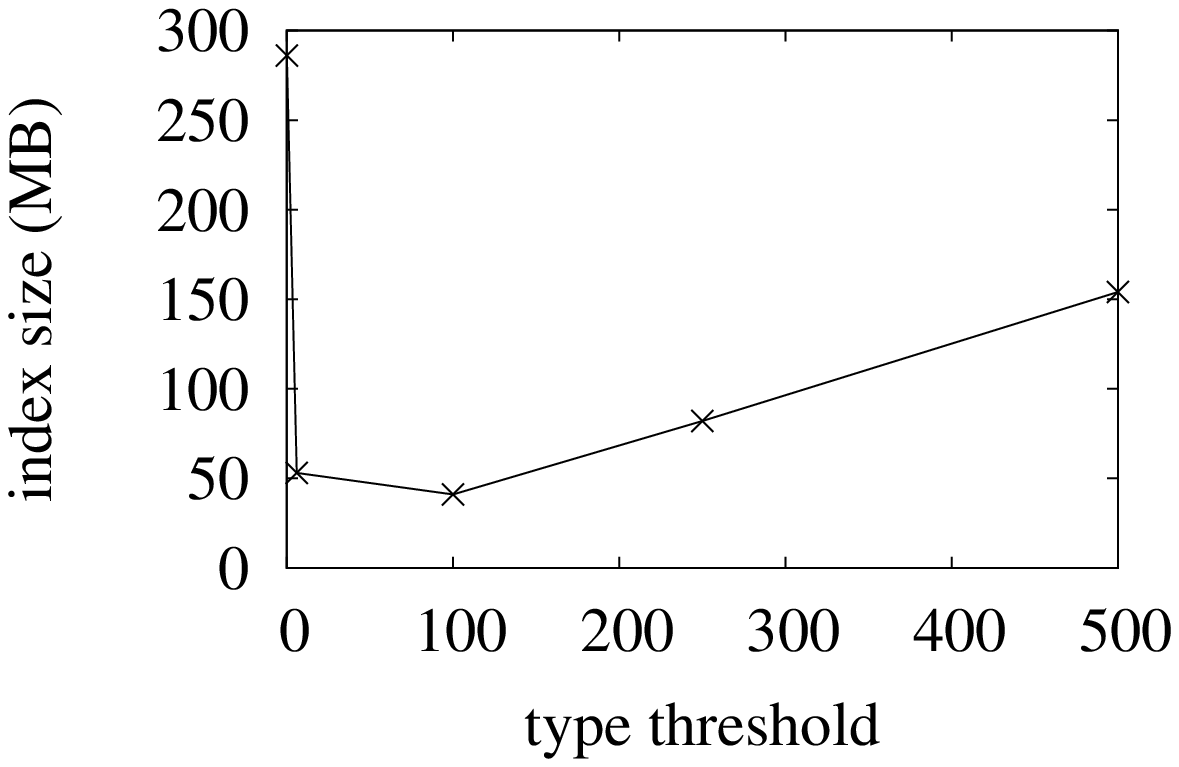}
}
\subfigure[n vs. search time]{
\includegraphics[height=32mm, width=39mm]{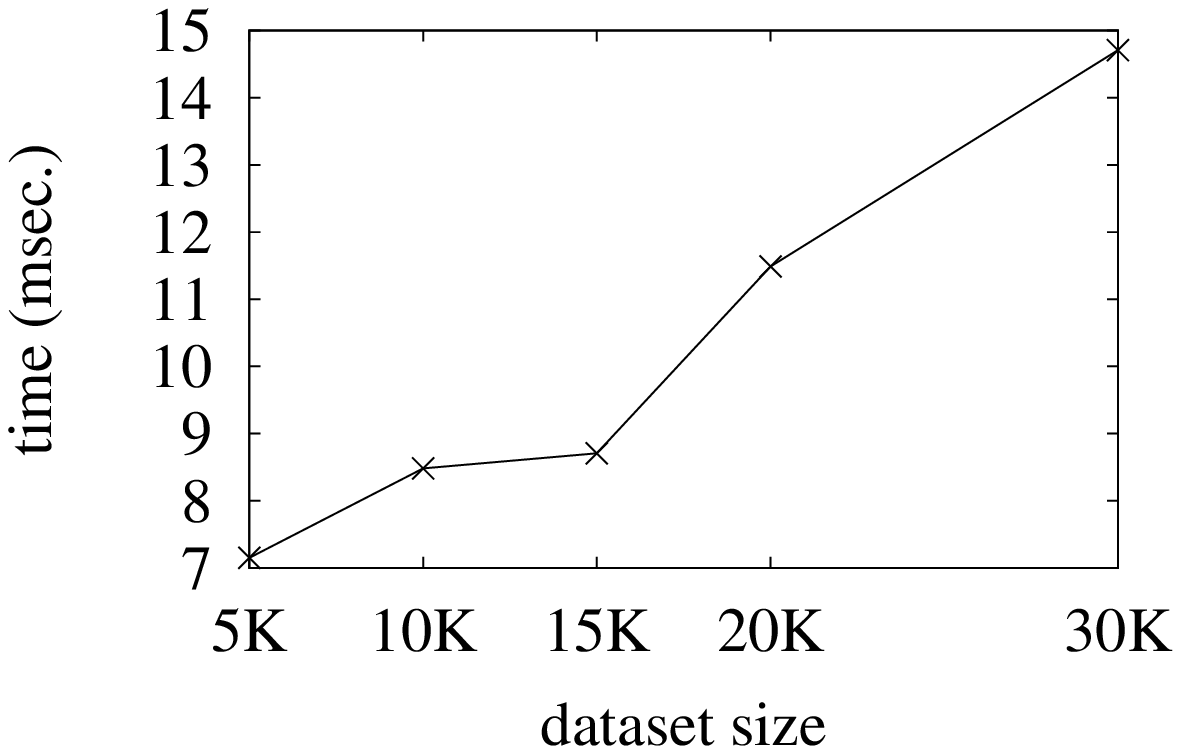}
}
\subfigure[m vs. search time]{
\includegraphics[height=32mm, width=39mm]{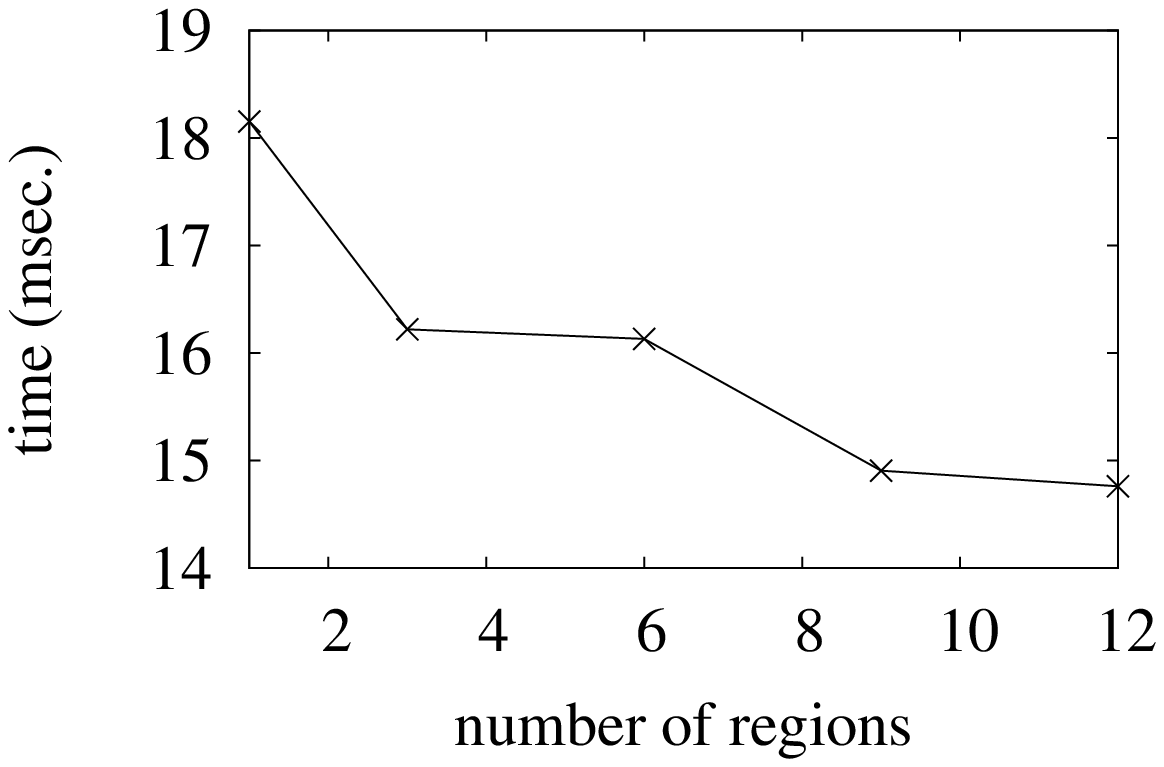}
}
\subfigure[|G(U)| vs. search time]{
\includegraphics[height=32mm, width=39mm]{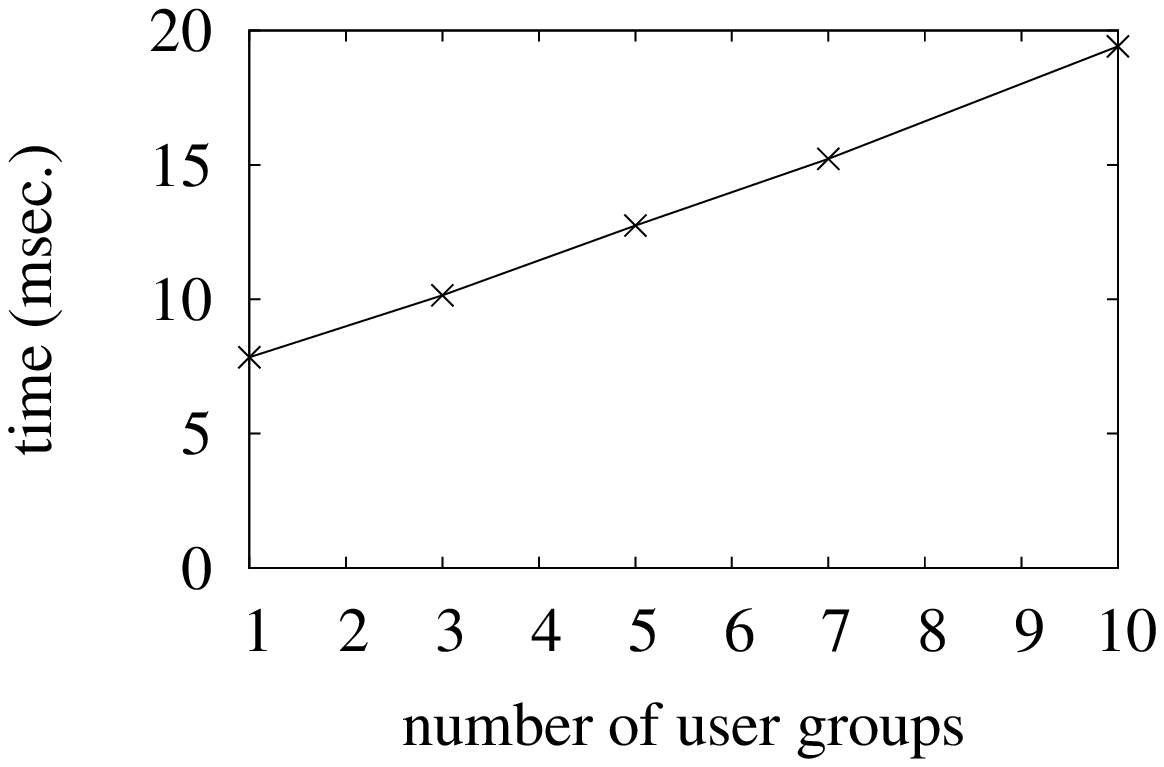}
}
\vspace{-3mm}
\caption[]{Influence of Protocol Parameters on Search Performance}
\label{Fig:Exp1}
\end{figure*}

To investigate the effect of distinct parameters on the search performance, we measured the time between the query request and identification of search results on the HBase cluster. To do so, we generated 1000 trapdoors by randomly drawing a keyword among all keywords. Reported timing results are the averages for the issued trapdoors. The search time is almost linear in the number of documents as depicted in Figure \ref{Fig:Exp1}-b. This is because, the number of encrypted payload blocks is proportional to the number of documents in the corpus. During the search phase, these blocks are decrypted to identify the matching documents.

Figure \ref{Fig:Exp1}-c demonstrates the influence of the number of regions generated on the HBase cluster. Note that, secure index is uniformly distributed among each region. Later, search on this distributed index is performed in parallel on each region. Hence, increase in the number of regions decrease the unit load of each machine.

In this study, we proposed authorization-aware keyword search protocol. During its execution, users are requested to issue a trapdoor component for each group that they are involved in. Once these components are received by the server, decryption operation is performed for each of them separately. Hence, the necessary computations is linearly proportional to the number of groups that a user is involved in as depicted in Figure \ref{Fig:Exp1}-d. Note that, constructed index payloads contain owner blocks in addition to user blocks. This enables owner queries to be executed as if the owner is a member of a single group.

To evaluate the scalability of the proposed scheme, we replicated the emails in our corpus 40 times. After the replication, we generated secure index on a corpus of 1,204,360 emails which fit into approximately 2.25 GB of storage. This index was further distributed to multiple regions in our HBase cluster. 
During this construction, if pure bit vector payloads were utilized instead of the proposed mixture of bit-vector and list payloads, the amount of necessary storage would be approximately 11 GB. Figure \ref{Fig:Exp2} demonstrates the average search time for the issued 1000 random search requests. Search operation could be performed in milliseconds according to our analysis. Note that, increase in the number of regions does not always reduce the search time. There is a slight increase in search time from $m = 1$ to $m = 3$. This is due to the network latency during the collection of partial results from the individual machines. Master needs to gather outputs from each server where regions are hosted. On the other hand, more regions generally decrease the total search time since each machine finishes its job faster with less load. It is clear that proposed scheme along with the capacity of real cloud infrastructures enables highly scalable search capability over encrypted document collections. 

\begin{figure}[!htb] 
\begin{center}
\includegraphics[scale = 0.35]{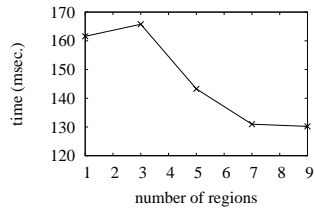}
\end{center}
\vspace{-8mm}
\caption{Scalability Evaluation}
\label{Fig:Exp2}
\end{figure}

\subsection{Update Evaluation}
In this part, we evaluate the computationally efficiency of the update mechanism. To perform our evaluation, we built secure update index on our email corpus as described in Section \ref{Update}. During its construction, we apply padding on the index payloads or document themselves to hide the actual number of keywords that are included in the documents based on a unit keyword-length parameter denoted as $avg_{|w|}$. Figure \ref{Fig:Exp3} demonstrates the influence of this parameter on the storage. Increase in $avg_{|w|}$  leads to a significant reduction in the index size. This is because, payload of the update index corresponding to document $D_i$ contains $|C_i| / avg_{|w|}$ entries where $C_i$ is the encrypted form of $D_i$. Hence, with increasing $avg_{|w|}$, number of entries in the payloads become less. On the other hand, increase in $avg_{|w|}$ leads to some increase in the size of the encrypted document collection. Note that, if $D_i$ contains $cnt_i$ keywords. Then $|C_i| / avg_{|w|} ~\ge~ cnt_i$ should hold true. If  $avg_{|w|}$ becomes larger, we need to pad documents to satisfy the necessary condition.

\begin{figure}[!htb] 
\begin{center}
\includegraphics[scale = 0.38]{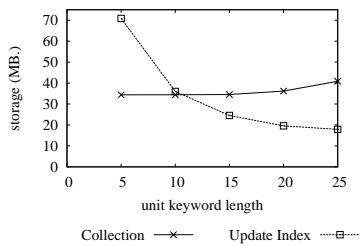}
\end{center}
\vspace{-6mm}
\caption{Influence of ${\mathbf{avg_{|w|}}}$ on Storage}
\label{Fig:Exp3}
\end{figure}

To evaluate the efficiency of the update operations, we measure the time between the request and completion of modifications on the indices. Resource consumption of both addition and deletion requests is based on the number of keywords that are associated with the operation. Hence, we generated update requests for documents with distinct amount of keywords. Specifically, we selected 10 random emails from the corpus for each distinct keyword size (e.g., 50, ..., 250) and we issued a token for deletion and addition of the selected emails in sequence. Reported timing results are the averages over 10 executions. Figure \ref{Fig:Exp4} depicts the update time for deletion and addition of a document with distinct keyword size. With increasing number of keywords, update time for both addition and deletion increases since the number of operations both on the search and update index are linearly proportional to the number of keywords.

\begin{figure}[!htb] 
\begin{center}
\includegraphics[scale = 0.38]{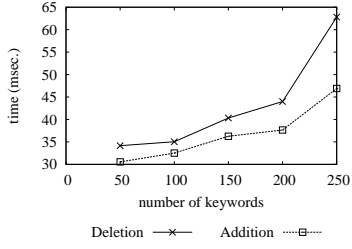}
\end{center}
\vspace{-6mm}
\caption{Update Efficieny}
\label{Fig:Exp4}
\end{figure}

\section{Conclusion}
\label{Conclusion}
In this paper, we propose a search scheme over encrypted documents for real cloud infrastructures. Proposed design is based on a distributed secure index which allows parallel execution of the search process among many machines. To restrict information leakage of the scheme according to the user access rights, proposed approach integrates authorization into the design. In addition, we also propose an effective update mechanism for distributed search index to be compatible with the dynamic nature of real storage systems. To investigate the efficiency of the proposed scheme, we conduct empirical analysis on a real dataset. Empirical evaluations indicate the practical nature of the proposed scheme.

In the context of this study, we provide a vertically partitioned index architecture to enable simultaneous decryption of large payloads by multiple machines during the search. In future work, we plan to design a hybrid architecture that consists of both vertically and horizontally partitioned segments to utilize resources according to query workload.

\appendix

\section{Security Definition}

Prior to the security definition, we need some auxiliary notions which can be summarized as follows:

\textbf{History $\mathbf{(H_\eta)}$:} Let D be the document collection and  $OP$ =$\{o_1, ... , o_\eta\}$ be the consecutive search, addition or deletion requests that are issued by users $U_1$, ... $U_\eta$ with respective access groups $G$ = \{$G(U_1)$, ..., $G(U_\eta)$\}. Then, $H_\eta = (D, OP, G)$ is defined as an $\eta$-query history.

\textbf{Trace ($\mathbf{\gamma}$):} Let $C = \{C_1, ..., C_n\}$ be the collection of encrypted documents, $|C_i|$ be the size of $C_i$, $g: D \mapsto G$ is a function that maps the document identifiers to the groups in $G$, $KW_p({H_\eta})$, $PT_p({H_\eta})$,  $A_p^g({H_\eta})$, $\beta_p(H_\eta)$, $\alpha_p(H_\eta)$ be the keyword, payload type, group-oriented access, deletion and addition patterns for $H_\eta$, $n_L$ and $n_B$ denotes the number of keywords in the search index with list and bit-vector payloads during the initial setup, $avg_{|w|}$ be the unit keyword length parameter for the construction of update index. Then, $\gamma(H_\eta)$ = \{$(|C_1|, ..., |C_n|)$, $g$, $KW_p(H_\eta)$, $PT_p(H_\eta)$, $A_p^g(H_\eta)$, $\beta_p(H_\eta)$, $\alpha_p(H_\eta)$ $n_L$, $n_B$, $avg_{|w|}$\} is defined as the trace of $H_\eta$. Trace is the maximum amount of information that a data owner allows its leakage to an adversary.

\textbf{View (v):} Let $C = \{C_1, ..., C_n\}$ be the collection of encrypted data items, $I$ be the secure search index, $U$ be the secure update index and $T = \{T_{o_1}, ..., T_{o_\eta}\}$ be the tokens for $\eta$ consecutive requests in $H_\eta$. Then, $v(H_\eta) = \{C, I, U, T\}$  is defined as the view of $H_\eta$. View is the information that is accessible to an adversary.

Now we can move into the security definition.

\textbf{Authorization-Aware Adaptive Semantic Security for Dynamic SSE:} SSE scheme satisfies authorization-aware adaptive semantic security in random oracle model, if there exists a probabilistic polynomial time simulator $S$ that can adaptively simulate the adversary's view of the history from the trace with probability negligibly close to $1$ through interaction with random oracle. Intuitively, this definition implies that all the information that is accessible to the adversary can be constructed from the trace. More formally, let $H_\eta$ be a random history from all possible histories, $v(H_\eta)$ be the view, $\gamma(H_\eta)$ be the trace of $H_\eta$. Then, scheme satisfies the security definition in random oracle model if one can define a simulator S such that for all polynomial size distinguishers $Dist$, for all polynomials $poly$ and a large $\theta$:   

\vspace{-3mm}
{\footnotesize{
\begin{equation*}
Pr[Dist(v(H_\eta)) = 1] - Pr[Dist(S(\gamma(H_\eta))) = 1] < \frac{1}{poly(\theta)} 
\end{equation*}}}

\noindent
where probabilities are taken over $H_\eta$ and the internal coins of key generation and encryption.

\begin{thrm}
Proposed scheme satisfies the authorization-aware adaptive semantic security.
\end{thrm}

{\emph {Proof.}} We will show the existence of polynomial size simulator $S$ such that the simulated view $v_S(H_\eta)$ and the real view $v_R(H_\eta)$ of history $H_\eta$ are computationally indistinguishable. Let $v_R(H_\eta)$ = \{$C$, $I$, $U$, $T$\} be the real view. Then, $S$ adaptively generates the simulated view $v_S(H_\eta)$ = \{$C^*$, $I^*$, $U^*$, $T^*$\} using trace $\gamma(H_\eta)$.

\begin{list}{\labelitemi}{\leftmargin=0.2em}

\item $S$ chooses n random values $\{{C_1}^*, ..., {C_n}^*\}$ such that $|{C_1}^*| = |C_1|, ...,  |{C_n}^*| = |C_n|$. In this setting, $C_i$ is output of a secure encryption scheme. By the pseudorandomness of the applied encryption, $C_i$ is computationally indistinguishable from ${C_i}^*$.

\item Given the number of documents and keywords in the collection (i.e., $n$, $n_L + n_B$) and identifier length (i.e., $|id|$), $S$ computes the number of elements for index payloads of type list which is represented as $\Upsilon$. Then it finds the number of blocks for list and bit-vector payload types which are denoted by $cnt_L$ and $cnt_B$ respectively. Specifically, $cnt_L$ is $\lceil \Upsilon / (\lfloor \kappa / |id| \rfloor) \rceil$ where $\kappa$ is output length of the random oracle $O_S$. Similarly, $cnt_B$ is $\lceil n /\kappa  \rceil$. In this setting, suppose $\varphi$ is the output length of pseudo-random function $\Phi$, $n_L$ and $n_B$ are number of keywords in the search index with list and bit-vector payloads. Then, $n_B$ index entries generated for keywords with bit-vector and $n_L$ entires are formed for keywords with list payloads. Specifically, $S$ generates pair ($k_i^{*}$, $v_i^{*}$) where $k_i^{*}$ is a random value of length $\varphi$ and $v_i^{*}$ is a collection of random values each of which has length $\kappa$. The number of elements in $v_i^{*}$  is $cnt_B$ for bit-vector and $cnt_L$ for list payloads. Finally, pairs $(k_i^{*}$, $v_i^{*})$ for $1 \le i \le (n_L + n_B)$ constitutes the simulated search index $I^{*}$. Note that, for any real entry ($k_i$, $v_i$) $\in I$, there is a corresponding simulated entry ($k_j^*$, $v_j^*$) $\in I^*$. Here length of $k_i$ and $k_j^*$, number of blocks in $v_i$ and $v_j^*$ along with the individual block lengths are equal. In addition, encrypted keys and blocks are computationally indistinguishable from random values by the pseudo-randomness of the applied encryptions. 

\item Given the length of ciphertexts \{$|C_1|$, ..., $|C_n|$\} along with the unit keyword length denoted as $avg_{|w|}$, S computes the number of payload entries for each document in the update index. Specifically, $u_i$ entries is formed for document $D_i$ where $p_i = |C_i|/avg_{|w|}$.
In this setting, suppose $\upsilon$ is the output length of pseudorandom function that is applied on document names, $\xi$ is the output length of random oracle $O_D$. Then $S$ generates a pair  ($Uk_i^{*}$, $Uv_i^{*}$) for each $C_i$ such that  $Uk_i^{*}$ is set to a random value of length $\upsilon$ and $Uv_i^{*}$ is set to a collection of $p_i$ random values each of which has length $\xi$. Finally, pairs  ($Uk_i^{*}$, $Uv_i^{*}$) for $1 \le i \le n$ constitutes the simulated update index $U^{*}$.
Note that, for real entry ($Uk_i$, $Uv_i$) $\in U$, there is a corresponding simulated entry ($Uk_i^*$, $Uv_i^*$) $\in U^*$. Here length of $k_i$ and $k_j^*$, number of blocks in $Uv_i$ and $Uv_i^*$ along with the individual block lengths are equal. In addition, encrypted key, value pairs are computationally indistinguishable from the random values by the pseudo-randomness of the applied encryption.

\item $S$ simulates requests $T_{o_1}$, ..., $T_{o_\eta}$ according to their type: 

{\bf 1) $\mathbf{T_{o_i}}$ is a search request:} Suppose $T_{o_i}$ = ($\pi_{i_1}$, \{[$G_{i_1}$, $K_{i_1}$], ..., [$G_{i_f}$, $K_{i_f}$]\})is a search request. Then, if $KW_p$($(i, 1)$, $(j,\ell)$) $= 1$ for any $1 \le j < i$, then $\pi_{i_1}^* = \pi_{j_\ell}^*$. Otherwise $\pi_i^*$ is set to a random row-key ${k_\iota}^*$ from the simulated search index in such a way that selected row-key was not previously selected during the simulation. Specifically, a row-key ${k_\iota}^*$ is selected from $I^*$ with a payload of type list if $PT_p(i, 1) = 1$ and a payload type of bit-vector otherwise.  Here, suppose $G(U_i)$ represents the access groups of the user who issued the request. If $KW_p((i,1), (j,1)) = 1$  and $G_{i_\rho} = G_{j_\vartheta}$  where $G_{i_\rho} \in G(U_i)$ and $G_{j_\vartheta} \in G(U_j)$ for any $1 \le j < i$, then $K_{i_\rho}^*$ = $K_{j_\vartheta}^*$. Otherwise  $K_{i_\rho}^*$ is set to a random value, length of which is equal to the output length of pseudo-random function $\Psi$. Note that group information that are associated with the search request is included in the group oriented access pattern (i.e., $A_p^g$). In this setting, components of simulated and real requests are computationally indistinguishable by the pseudo-randomness of the applied encryption. Hence, $T_{o_i}$ and $T_{o_i}^*$ are computationally indistinguishable. To ensure that server observes the same data access against $T_{o_i}$ and $T_{o_i}^*$, $S$ programs  random oracle $O_S$ according to $A_p^g$. Suppose $k_{\iota}^*$ is row-key in $I^*$ that is assigned to $\pi_{i_1}^*$ component of $T{o_i}^*$. Then, for each $A_p(w_x, G_{i_\rho})$ that is included in $A_p^g(o_i)$, $S$ selects blocks from $v_{\iota}^*$ which is the value corresponding to $k_{\iota}^*$ in $I^*$. Here, suppose [$id_j$, $block(id_j)$, $slot(id_j)$] $\in$  $A_p(w_x, G_{i_\rho})$. Then S selects $k^{th}$ block from  $v_{\iota}^*$ where $k = block(id_j)$ and  programs $O_S$ in such a way that when the slot with order $slot(id_j)$ in this block is decrypted with key $K_{i_\rho}^*$, server observes $id_j$.

{\bf 2) $\mathbf{T_{o_i}}$ is  a deletion request:} Suppose $T_{o_i}$ = ($\sigma_i$, $K_i$) is a deletion request, $\beta_p(o_i)$ = ($L_{o_i}$, $|C_i|$) is  deletion pattern for $o_i$, $\varrho$ is the output length of pseudo-random function that is used for payload encryption of the update index. Then $S$ selects previously unselected key-value pair ($Uk_j^*$, $Uv_j^*$) from simulated update index $U^*$ in such a way that  number of entries in  $Uv_j^*$ is $\lceil |C_i| / avg_{|w|} \rceil$. Then $S$ sets $\sigma_i^*$ to $Uk_j^*$ and $K_i^*$ to a random value of length $\varrho$. To ensure that adversary applies the correct modifications on the search index,  $S$ utilizes $\beta_p$ and $KW_p$. Suppose $L(o_i)$ = \{$block_{i_1}||slot_{i_1}$, ..., $block_{i_\upsilon}||slot_{i_\upsilon}$\}. Then, for each entry $block_{i_\rho}||slot_{i_\rho}$ in  $L(o_i)$, $S$ initially selects a row-key  $\pi_{i_\rho}^*$ from $I^*$. Specifically, if $KW_p((i,\rho), (j,\ell)) = 1$ for any $1 \le j < i$, then $\pi_{i_\rho}^* = \pi_{j_\ell}^*$. Otherwise $\pi_{i_\rho}^*$ is set to a random row-key from $I^*$ in the same way as row-key selection for the search request. Once  $\pi_{i_\rho}^*$ is fixed, S programs $O_D$ in such a way that, adversary observes 
$\pi_{i_\rho}^*$||$block_{i_\rho}||slot_{i_\rho}$ when $\rho^{th}$ entry of $Uv_j^*$ is decrypted with key $K_i^*$. Note that, adversary observes the correct addresses for the modifications on the search index once $O_D$ is programmed and it performs the necessary modifications on the search index according to protocol flow.

{\bf 3) $\mathbf{T_{o_i}}$ is an addition request:}  Suppose $T_{o_i}$=  (($\pi_{w_1^u}$, $cnt(w_1^u)$), ..., ($\pi_{w_\upsilon^u}$, $cnt(w_\upsilon^u)$) is a deletion request, $\alpha_p(o_i)$ = (\{$|C_\iota|$, ..., $|C_\rho|$\}, \{$cnt(w_1^u)$, ...,$cnt(w_\upsilon^u)$\} denotes addition pattern. Then, for each pair ($\pi_{w_\rho^u}$, $cnt(w_\rho^u)$), S simply copies $cnt(w_\rho^u)$ from $A_p(o_i)$ to form  $cnt(w_\rho^u)^*$. Then, if $KW_p((i,\rho), (j,\ell)) = 1$ for any $1 \le j < i$, then $\pi_i^* = \pi_{j_\ell}^*$. Otherwise $\pi_i^*$ is set to a random row-key from the simulated index $I^*$ as in the token construction for search requests. Note that addition is a two-round process. Once adversary receives the simulated addition request, it identifies the blocks with available slots in $I^*$ for each $\pi_{w_\rho^u}^*$ according to the protocol flow and returns them. Suppose $\Delta(w_\rho^u)^*$ is the block list that is received from the adversary for $\pi_{w_\rho^u}^*$. Each block in  $\pi_{w_\rho^u}^*$ is in the form of $v_j^*[\iota] || S[\iota]$ where $v_j^*[\iota]$ is a column value in $I^*$ and $S[\iota]$ is the corresponding signal array which indicates the invalid slots in $v_j^*[\iota]$. In this setting, suppose $av_\rho$ is the number of invalid slots in   $\Delta(w_\rho^u)^*$ and $slot_{cnt}$ is the number of slots in each block. Then S generates  $|\Delta(w_\rho^u)^*|$ + $\lceil((cnt(w_\rho^u) - av_\rho)/slot_{cnt}) \rceil$ new simulated blocks such that each block is in the form of $r_j^* || S_j^*$ where $r_j^*$ is a random value of length $|v_j^*[\iota]|$ and $S_j^*$ is a zero vector of length $slot_{cnt}$. Finally, these blocks are returned to adversary as simulated update blocks. During the addition simulation, $S$ needs to simulate new ciphertexts and update index entries as well.  Given the length of ciphertexts  \{$|C_\iota|$, ..., $|C_\rho|$\}, $S$ performs the ciphertext and update index simulation process of  steps 1 and 3 respectively.
\end{list}

\noindent
Since each component of $v_R(H_\eta)$ and  $v_S(H_\eta)$ are computationally indistinguishable, we can conclude that the proposed scheme satisfies the  security definition.

\bibliographystyle{acm}
\bibliography{bibliography}

\end{document}